\documentclass[a4paper]{article}

\usepackage{cite}
\usepackage{url}
\usepackage{amsmath}
\usepackage{amsfonts}
\usepackage{caption}
\usepackage{amsthm}

\usepackage{tikz}
\usetikzlibrary{arrows}
\tikzstyle{carre}=[draw, minimum size=2em]

\newcommand{\ff}{\mathbb{F}}
\newcommand{\ie}{\textit{i. e. }}

\newtheorem{theorem}{Theorem}

\title{Towards Secure Two-Party Computation from the Wire-Tap Channel}

\author{Herv\'e Chabanne$^{1,2}$, G\'erard Cohen$^{2}$, Alain Patey$^{1,2}$\\
$^1$Morpho, $^2$T\'el\'ecom ParisTech\\
Identity and Security Alliance\\ (The Morpho and T\'el\'ecom ParisTech Research Center)}
\date{}

\begin{document}

\maketitle

\begin{abstract}
We introduce a new protocol for secure two-party computation of linear functions in the semi-honest model, based on coding techniques. We first establish a parallel between the second version of the wire-tap channel model and secure two-party computation. This leads us to our protocol, that combines linear coset coding and oblivious transfer techniques. Our construction requires the use of binary intersecting codes or $q$-ary minimal codes, which are also studied in this paper.
\end{abstract}


\section{Introduction}

Secure Multi-party Computation has been introduced in the late eighties by Yao \cite{Yao86} and has been subject to a lot of studies to demonstrate its feasibility and completeness in several adversarial settings. Recently, a lot of work has been done to make these techniques practical. We refer the reader to \cite{Gol04,HL10,CDN12} for  overviews on the state of the art in Secure Multi-Party Computation. We here focus on the two-party setting. In this setting, two parties $P_1$ and $P_2$, holding respective inputs $X$ and $Y$, wish to securely compute a function $f$ on their inputs. At the end of the protocol, one party (or both) learns $f(X,Y)$, but gains no more information about the other party's input than what can be deduced from this output. The seminal example given by \cite{Yao86} is the millionaire's problem: two millionaires wish to know which one of them is the richer, without revealing their respective wealths. We here focus on the semi-honest adversarial model, where both parties are supposed to follow the protocol but where they try to infer more information than they should from all data exchanges. Yao \cite{Yao86} gives a construction fulfilling these requirements \cite{LP09}, applicable to any function expressed as binary circuit. This technique is based on garbled circuits and oblivious transfer.

Oblivious transfer, originally introduced by Rabin \cite{Rab81} in a slightly different version, enables one receiver $R$ to get one out of $N$ secrets $X_1,\ldots,X_N$ held by a sender $S$. The receiver chooses an index $c \in \{1,\ldots,N\}$, gets $X_c$ and learns nothing about the $X_j$'s, for $j \neq c$. Symmetrically, the sender $S$ learns nothing about $c$. This thus also known as \textit{Symmetric Private Information Retrieval} (SPIR). Many protocols and implementations exist for oblivious transfer, some pointers can be found in \cite{Lip}.

The Wire-Tap Channel model has been introduced by Wyner \cite{Wyn75} and later extended by Ozarow and Wyner \cite{OW84} to a second version considering an erasure channel for the eavesdropper. We here consider the Wire-Tap Channel II (WTC2) \cite{OW84} to establish a parallel with Secure Two-Party Computation. The model for WTC2 is described in Figure~\ref{fig:wtc}. Alice sends an encoded message to Bob. Eve is allowed to access a bounded number of coordinates of the codeword, and she moreover controls the erasure positions. In the original model, Eve is not supposed to learn any information about the original message, even knowing the coding and decoding algorithms. Later \cite{Wei91,CLZ94}, the information gained by Eve if she learns more than the original bound was studied. In particular, using  coset coding techniques, there exists a sequence $(d_i)$ of bounds such that Alice gains less than $i$ information bits about the original message if she has access to less than $d_i$ coordinates of the message.

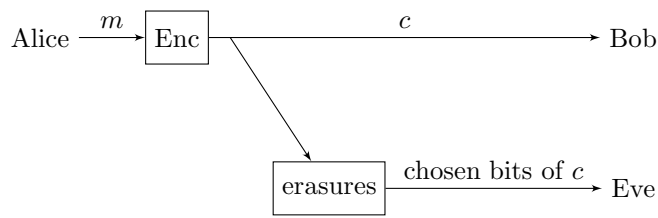
\begin{figure}[h!]
\begin{center}

\begin{tikzpicture}[transform shape,node distance=2cm,auto,>=latex']
  \draw
    (0,0) node [] (a) {Alice}
    node [right of=a, carre, node distance=1.8cm] (b) {Enc}
    node [right of=b, coordinate, node distance=.7cm] (begintap) {}
    node[right of = b, coordinate, node distance=2 cm](c){}
    node [below of=c, node distance=2cm, carre] (endtap) {erasures}
    node [right of=b,node distance=6cm] (d) {Bob}
    node [below of=d,node distance=2cm](e){Eve};

  \path[->] (a) edge node {$m$} (b);
  \path[->] (b) edge node {$c$} (d);
  \path[->] (begintap) edge node {} (endtap);
  \path[->](endtap)edge node {chosen bits of $c$}(e);
\end{tikzpicture}

\caption{The Wire-Tap Channel II}
\label{fig:wtc}
\end{center}
\end{figure}

This is where we establish the parallel with Secure Two-Party Computation. We see the two parties performing the secure computation as Alice and Eve in the WTC2 model. The message that is encoded by Alice would be the input $X$ of Alice. We want the bits of information that Eve gets about $X$ to be the actual bits of $f(X,Y)$. We will explain in this paper how to do this using linear coset coding techniques and some classes of linear functions. The last thing we need to achieve the parallel is a modeling of the erasure channel. This will be done using  oblivious transfers. We illustrate this parallel in Figure~\ref{fig:smcwtc}.

\begin{figure}[h!]
\begin{center}

\begin{tikzpicture}[transform shape,node distance=2cm,auto,>=latex', scale=.8]
  \draw
    (0,0) node [] (a) {$P_1$}
    node [right of=a, carre, node distance=1.8cm] (b) {Enc}
    node [right of=b, coordinate, node distance=.7cm] (begintap) {}
    node[right of = b, coordinate, node distance=2 cm](c){}
    node [below of=c, node distance=2cm, carre] (endtap) {oblivious transfers}
    node [below of=d,node distance=2cm](e){$P_2$};

  \path[->] (a) edge node {$X$} (b);
  \path[->] (b) edge node {$c$} (endtap);
  \path[->](endtap)edge node {chosen bits of $c$}(e);
\end{tikzpicture}

\caption{From WTC2 to Secure Two-Party Computation}
\label{fig:smcwtc}
\end{center}
\end{figure}
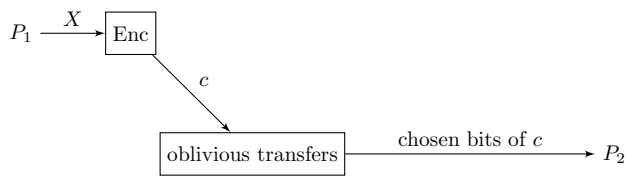

In Section~\ref{sec:wtc}, we recall some results about Wire-Tap Channel II and linear coset coding. We infer a protocol for secure two-party computation in Section~\ref{sec:protocol}. This raises the problem of finding minimal linear codes, that we study in Section~\ref{sec:minimal}. Finally, we conclude in Section~\ref{sec:conclu}.

\section{Wire-Tap Channel II and Linear Coset Coding}
\label{sec:wtc}

In the following, a \emph{$[n,k,d]$ linear code} denotes a subspace of dimension $k$ of $\ff_q^n$ with minimal Hamming distance $d$, where $q=p^k$, for $p$ prime and $k \in \mathbb{N}$. We denote by $C^{\bot}$ the dual code of $C$. The support of $c \in C$ is $supp(c)=\{i \in \{1,\ldots,n\} | c_i \neq 0\}$. We might use \textit{bit}, by abuse of language, even if $q \neq 2$, to denote a coordinate of a message or of a codeword.

\subsection{Linear Coset Coding}

Coset coding is a random encoding used for both models of Wire-Tap Channel \cite{Wyn75,OW84}. This type of encoding uses a $[n,k,d]$ linear code $C$ with a parity-check matrix $H$. Let $r = n-k$. To encode a message $m \in \mathbb{F}_q^r$, one  randomly chooses an element among all $x \in \mathbb{F}_q^n$ such that $m =H ^t \! x$. To decode a codeword $x$, one just applies the parity-check matrix $H$ and obtains the syndrome of $x$ for the code $C$, which is the message $m$. This procedure is summed up in Figure~\ref{fig:coset}.

\begin{figure}[h!]

\centering
\fbox{
\begin{minipage}{0.9\linewidth}
Given: $C$ a $[n,n-r,d]$ linear code with a $r \times n$ parity-check matrix $H$\\
\textbf{Encode}: $m \in \mathbb{F}_2^r \mapsto_R x \in \mathbb{F}_2^n$ s.t. $H ^t \! x = m$ \\ 
\textbf{Decode}: $x \in \mathbb{F}_2^n \mapsto m=H ^t \! x$
\end{minipage}
}
\caption{Linear Coset-coding}
\label{fig:coset}
\end{figure}

\subsection{The Wire-Tap Channel I}

The Wire-Tap Channel was introduced by Wyner \cite{Wyn75}. In this model, a sender Alice sends messages over a potentially noisy channel to a receiver Bob. An adversary Eve listens to an auxiliary channel, the Wire-Tap channel, which is a noisier version of the main channel. It was shown that, with an appropriate coding scheme, the secret message can be conveyed in such a way that Bob has complete knowledge of the secret and Eve does not learn anything. In the special case where the main channel is noiseless, the secrecy capacity can be achieved through a linear coset coding scheme.

%
%
%
%

\subsection{The Wire-Tap Channel II}

Ten years later, Ozarow and Wyner introduced a second version of the WT Channel \cite{OW84}. In this model, both main and Wire-Tap channels are noiseless. This time, the disadvantage for Eve is that she can only see messages with erasures: she has only access to a limited number of bits per codeword. She is however allowed to choose which bits she can learn. We summarize the Wire-Tap Chanel II in Figure~\ref{fig:wtc}.

The encoding used in this model is again a coset coding based on a linear code $C$, as in the Wire Tap Channel I with a noiseless main channel. Let $d^{\bot}$ denote the minimal distance of the dual $C^{\bot}$ of $C$. One can prove (see \cite{Wei91} for instance) that, if Eve can access less than $d^{\bot}$ bits of a codeword, then she gains no information at all on the associated message.


\subsection{Generalized Hamming Distances}

Generalized Hamming distances (or generalized Hamming weights) have first been considered by Wei \cite{Wei91}. The $i^{th}$ generalized Hamming distance, denoted by $d_i (C)$ or $d_i$ is the minimum size of the union of the supports of $i$ linearly independent codewords in $C$. We have $1 \leq d=d_1 \leq \ldots \leq d_k \leq n$.

Using generalized Hamming distances, we get a more precise evaluation of the information gained by Eve in the WTC2, depending on the linear code used for coset coding. For $i=1,\ldots,r$, let $d^{\bot}_i$ denote the $i^{th}$ generalized Hamming distance of $C^{\bot}$, the dual code of $C$. We have the following result \cite{Wei91}:

\begin{theorem}[WTC2 and Generalized Hamming Distances]
If Eve gets less than $d^{\bot}_i$ bits of the codeword $c$, she gains at most $i-1$ information bits about the original message $m$.
\end{theorem}

\section{Our Protocol for Secure Two-Party Computation}
\label{sec:protocol}

\subsection{The Setting}

We describe our setting in Figure~\ref{fig:smc}. Notice that we can also give the result to $P_1$: since we work in the semi-honest model, where both parties follow the protocol, we can let $P_2$ send $f(X,Y)$ to $P_1$, once he has computed it.

\begin{figure}[htbp!]
                \centering

                \fbox{
      \begin{minipage}[c]{0.9\linewidth}
      
      \textbf{Inputs:}
      
$\bullet$      Party $P_1$ inputs $X \in \ff_q^r$

$\bullet$       Party $P_2$ inputs $Y \in S$
      
$\bullet$ Both parties know a description of $f: \ff_q^r \times S \rightarrow \ff_q$      
      
      \textbf{Outputs:}
      
$\bullet$		$P_1$ learns nothing about $Y$

$\bullet$		 $P_2$ obtains $f(X,Y)$ but learns nothing more about $X$ than what can be inferred from $f(X,Y)$.


 \end{minipage}
      }
\caption{Our Secure Two-Party Computation Setting}\label{fig:smc}
\end{figure}

We consider the secure evaluation of functions of the form 
\begin{eqnarray*}
f : & \ff_q^r \times S & \rightarrow  \ff_q \\
& (X,Y) & \mapsto  f(X,Y) = \sum\limits_{i=1}^r f_i(Y) \cdot x_i
\end{eqnarray*}
where $S$ is a given set, and $f_i : S \rightarrow \ff_q$, for $i=1,\ldots,r$. This class covers all linear functions of $X$ and $Y$ with range $\ff_q$ (\ie giving one "bit of information" about $X$ to $P_2$). 

For instance, if $Y \in \ff_q^r$ and $f_i(Y)=y_i$, $f$ is the scalar product over $\ff_q^r$. 

Squared euclidean distance can also be computed this way. $P_1$ also inputs $x_{r+1} = \sum\limits_{i=1}^r x_i^2$ and $f_i(Y)=-2y_i$, for $i=1,\ldots,r$, $f_{r+1}(Y)=1$. Thus , $P_2$ obtains $\sum\limits_{i=1}^r x_i^2 - 2 x_i y_i$, which is equivalent (for $P_2$) to the knowledge of $d(X,Y) = \sum\limits_{i=1}^r (x_i-y_i)^2$: it gives no additional information. 

If $q=p > \log(r)$ and inputs are binary vectors seen in $\ff_q$, it is also possible to compute Hamming distance (take $f_i(Y)=1-2y_i$). 

Securely computing these functions has applications in the signal processing and cryptographic domains, especially for privacy-preserving biometric recognition \cite{SSW09,BCP13}.

\subsection{From the Wire-Tap Channel to Secure Two-Party Computation}

As discussed in the introduction and illustrated in Figure~\ref{fig:smcwtc}, we transpose the WTC2 model to the Secure Two-party Computation setting, by assigning the role of Alice to $P_1$, the role of Eve to $P_2$ and modelling the erasure channel by oblivious transfers. We will use the notation $OT_t^n$ to denote the $t$-out-of-$n$ functionality described in Figure~\ref{fig:ot}. This can be implemented either using $t$ $OT_1^n$'s or more specific constructions, see \cite{Lip}.

\begin{figure}[htbp!]
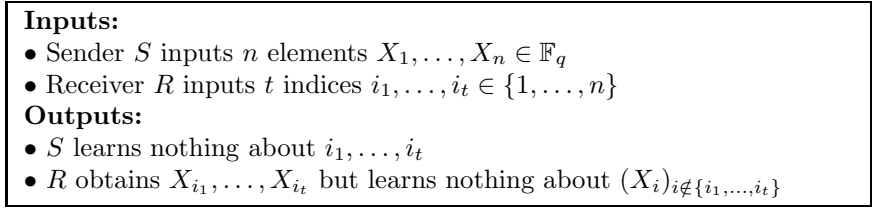

                \centering

                \fbox{
      \begin{minipage}[c]{0.9\linewidth}
      
      \textbf{Inputs:}
      
$\bullet$      Sender $S$ inputs $n$ elements  $X_1,\ldots,X_n \in \ff_q$

$\bullet$       Receiver $R$ inputs $t$ indices $i_1,\ldots,i_t \in \{1,\ldots,n\}$
      
      \textbf{Outputs:}
      
$\bullet$		$S$ learns nothing about $i_1,\ldots,i_t$

$\bullet$		 $R$ obtains $X_{i_1},\ldots,X_{i_t}$ but learns nothing about $(X_i)_{i \notin \{i_1,\ldots,i_t\}}$


 \end{minipage}
      }
\caption{The $OT_t^n$ Functionality}\label{fig:ot}
\end{figure}

\subsection{Choosing the Code}
\label{sec:code}

Let us first see how $P_2$ can choose the coordinates of the codeword that he gets through oblivious transfer, in order to obtain $f(X,Y)$. Let us consider the $r \times n$ matrix $H$ that is the parity-check matrix of the code $C$ used for coset coding, or, equivalently, the generator matrix of its dual code $c^{\bot}$. We denote by $H_i$ the $i^{th}$ row of $H$. Let $Z$ be an encoding of $X$, \ie such that $X =H ^t \! Z = \sum H_i z_i$. We consequently have $x_i = H_i \cdot ^t \!Z$ and $f(X,Y) = \sum f_i(Y) \cdot x_i = \sum f_i (Y) \cdot H_i \cdot ^t \!Z = (\sum f_i(Y)\cdot H_i) \cdot ^t \!Z$.

Thus, $P_2$ only needs the coordinates of $Z$ at the positions where $\sum f_i(Y) \cdot H_i$ is nonzero, \ie at the positions belonging to the support of $V=\sum f_i(Y) \cdot H_i$. This will ensure correctness. Let $i_1,\ldots,i_t = supp(V)$.

Now we need to ensure privacy of $P_1$'s data. We assume that $P_2$ only gets $z_{i_1},\ldots,z_{i_t}$. If there exists another vector $W \in C^{\bot}$, such that $V$ and $W$ are linearly independent and $supp(W) \subset supp(V)$, then $P_1$ learns at least another bit of information ($W ^t \! Z$)  about $Z$. To ensure $P_1$ only learns $f(X,Y)$, we need to enforce that $V$ is minimal in $C^{\bot}$, \ie that his support does not contain the support of another linearly independent codeword $W \in C^{\bot}$. Since we wish to ensure a notion of completeness, \ie to make our protocol usable with any $f$ and $Y$ fitting our setting, we require every codeword of $C^{\bot}$ to be minimal, \ie we require $C^{\bot}$ to be a minimal linear code (see Section~\ref{sec:minimal}).

Now let us fix some $V \in C^{\bot}$, let $t=|supp(V)|$ and let us consider the linear application $\phi : C^{\bot} \rightarrow \ff_q^{n-t} ; c \mapsto (c_i)_{i \notin supp(V)}$. Due to the definition of linearity, only the $\lambda V$, for $\lambda \in \ff_q$ have a support included in $supp(V)$ thus $Ker \phi = \ff_q . V$ and $rank(\phi)=dim(C^{\bot})-1=k-1$. Thus, if we let $P_2$ learn the $t$ coordinates of $Z$ corresponding to $supp(V)$, the remaining coordinates lie in a space of dimension $k-1$ and $P_2$ only learns one bit of information about $X$.

Consequently, using a minimal codeword ensures privacy of $P_1$ against $P_2$.

\subsection{Our Protocol}

We put together our studies of the last paragraphs and we get the protocol described in Figure~\ref{fig:prot}. Privacy against $P_2$ is ensured thanks to the remarks of Section~\ref{sec:code} and privacy against $P_1$ is ensured by the use of oblivious transfer, which is the only data exchange from $P_2$ to $P_1$. Correctness is also discussed in Section~\ref{sec:code}.

Some details still need to be considered. The size $t$ of $supp(V)$ can reveal information about $Y$ to $P_1$. Thus, either we need an oblivious transfer protocol that hides to the sender the number of transferred items, or we require $P_2$ to perform $w_{max} -t$ dummy requests, where $w_{max}$ is the maximal weight of a codeword of $C^{\bot}$. Since we work in the semi-honest model, this will not break the security properties (of course, a malicious (active) adversary would use real requests instead, but that setting is out of the scope of this paper).

\begin{figure}[htbp!]
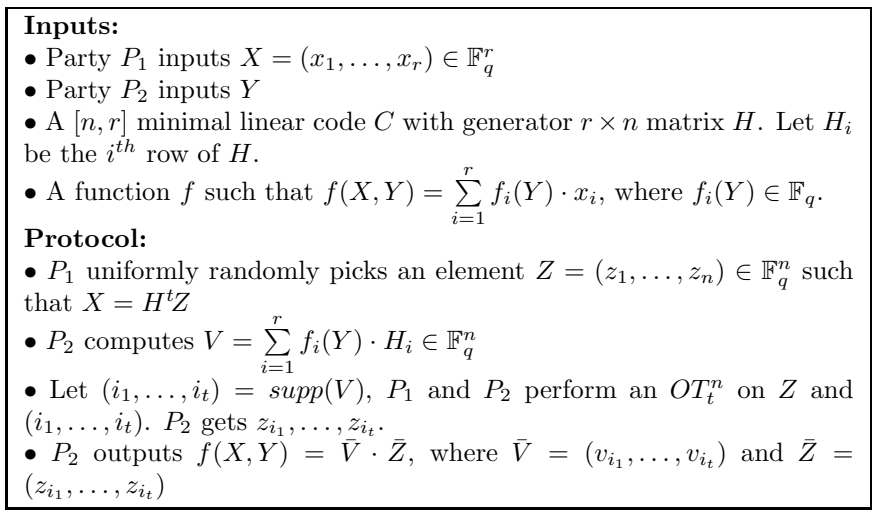

                \centering

                \fbox{
      \begin{minipage}[c]{0.9\linewidth}
      
      \textbf{Inputs:}
      
$\bullet$      Party $P_1$ inputs $X=(x_1,\ldots,x_r) \in \ff_q^r$

$\bullet$       Party $P_2$ inputs $Y$
      
$\bullet$       A $[n,r]$ minimal linear code $C$ with generator $r \times n$ matrix $H$. Let $H_i$ be the $i^{th}$ row of $H$.

$\bullet$ A function $f$ such that $f(X,Y) = \sum\limits_{i=1}^r f_i(Y) \cdot x_i$, where $f_i(Y) \in \ff_q$.

	\textbf{Protocol:}
	
	$\bullet$ $P_1$ uniformly randomly picks an element $Z=(z_1,\ldots,z_n) \in \ff_q^n$ such that $X =H ^t \! Z$
	
	$\bullet$ $P_2$ computes $V=\sum\limits_{i=1}^r f_i(Y) \cdot H_i   \in \ff_q^n$
	
	$\bullet$ Let $(i_1,\ldots,i_t)=supp(V)$, $P_1$ and $P_2$ perform an $OT_t^n$ on $Z$ and $(i_1,\ldots,i_t)$. $P_2$ gets $z_{i_1},\ldots,z_{i_t}$.
	
	$\bullet$ $P_2$ outputs $f(X,Y) = \bar{V} \cdot \bar{Z}$, where $\bar{V}=(v_{i_1},\ldots,v_{i_t})$ and 
	$\bar{Z}=(z_{i_1},\ldots,z_{i_t})$
     
%
%


 \end{minipage}
      }
\caption{Our Protocol for Secure Two-Party Computation}\label{fig:prot}
\end{figure}

We would like to point out that this protocol might not only have theoretical interest. For instance, the protocol of \cite{BCP13} uses coding-like techniques and oblivious transfer only, and is one of the most efficient protocol for securely computing functions such as Hamming distances on binary vectors, outperforming protocols based on additively homomorphic cryptosystems or on garbled circuits. In the case of the protocol of this paper, performance will highly rely on the rate of the underlying code. As we explain in Section~\ref{sec:minimal}, we are lacking results in the $q$-ary case.

\subsection{Examples}

We consider as an illustration the secure evaluation of scalar product over $\ff_q^r$, \ie $f(X,Y)=\sum\limits_{i=1}^r x_i \cdot y_i$. One can deduce how to proceed for any function encompassed by our protocol, by replacing $y_i$ by $f_i(Y)$.

\subsubsection{Simplex and Hamming Codes}

One can easily be convinced that constant-weight codes are minimal, in the binary or the $q$-ary case. Since we use linear codes, constant-weight codes are simplex codes (or equivalent), duals of Hamming codes \cite{Bon84}. Let $q=2,r=3, n=7$. The $3 \times 7$ matrix $H$ can for example be written as follows:
\[H=
\begin{pmatrix}
0 &0 &0&1&1&1&1\\
0&1&1&0&0&1&1\\
1&0&1&0&1&0&1
\end{pmatrix}
\]

Let $X=(101)$ and $Y=(110)$. $P_1$ can for instance encode $X$ with $Z=(0000100)$. $Y$ computes $V=H_1 + H_2 = (0111100)$ and requests, using oblivious transfers, the bits $z_2,z_3,z_4,z_5$. $P_2$ thus gets $\bar{Z}=(0001)$. By dot-product with $\bar{V}$, $P_2$ gets the result $f(X,Y)=\sum x_i \cdot y_i = 1$.

Notice that, since the code is constant-weight, $P_2$ always requests 4 bits, we thus do not need to hide the number of requested bits. Unfortunately, this nice property is only enjoyed by simplex codes, that have a very bad rate, $n$ growing exponentially with $r$, the rates being even worse in the $q$-ary case.

\subsubsection{A More Efficient Binary Example}
\label{sec:example}

In the binary case, we can easily obtain minimal codes with better rates than simplex codes (see Section~\ref{sec:minimal}). For instance, let $r=4$, we can have $n=9$ (optimal \cite{Slo93}), for instance using
\[
H=
\begin{pmatrix}
1&0&1&0&0&0&1&0&1\\
0&1&1&0&0&0&0&1&1\\
0&0&0&1&0&1&1&0&1\\
0&0&0&0&1&1&0&1&1
\end{pmatrix}
\]

Using this code, $P_2$ will request either 4 or 6 coordinates of $Z$, to obtain $f(X,Y)$, depending on $Y$. For instance if $Y=(1000)$ or $Y=(0011)$, $P_2$ will only request 4 coordinates, but if $Y=(0110)$, $P_2$ will need 6 coordinates.

\subsubsection{Comparison to the Yao's Protocol}

Let us consider secure evaluation of scalar product over $\ff_2^r$ using Yao's protocol \cite{Yao86,HL10,Sch12}. The binary circuit contains $r$ AND gates, we do not count XOR gates (see \cite{Sch12} and references therein for known optimizations on garbled circuits). Let $k$ be a security parameter (e.g. 80 or 128). Party $P_1$ has to compute $r$ garbled gates (4$r$ hash function evaluations). Party $P_2$ has to evaluate $r$ garbled gates gates ($r$ hash function evaluations). They perform $k$ $OT_1^2$'s on $k$-bit inputs ($P_2$'s input wire labels). Furthermore, $P_1$ also needs to send $r$ $k$-bit keys ($P_1$'s input wire labels) and $r$ garbled gates ($3 r k$ bits).

Now let us consider our protocol using a $[n,r]$ minimal code with maximum codeword Hamming weight equal to $w_{max}$. Our protocol requires linear algebra operations and a $OT_{w_{max}}^n$, with 1-bit inputs. For instance, the $OT_{w_{max}}^n$ operation can be realized using $w_{max} OT_1^{n-w_{max}+1}$, still with 1-bit inputs, but there might be more efficient procedures. Using for instance the construction of \cite{CZ94} to build minimal binary codes, one can have $n \approx 6.4 r$, for any $r$. This comparison in the binary case is summed up in Table~\ref{tab:comp}.

\begin{table}

\begin{center}

\begin{tabular}{|c|c|c|c|c|}
\hline
Protocol & OT (computation  & Add. data  & Add. compu-  &Add. compu-\\
& + data exchanges)&exchanges&tation ($P_1$)&tation ($P_2$) \\
\hline
Yao & $r \times OT_1^2$ & $4 r k$ bits & $4 r$ hash  & $r$ hash \\
& ($k$-bit inputs)&&function eval.&function eval.\\
\hline
Our  & $1 \times OT_{w_{max}}^n$  & $\emptyset$ & Linear  & Linear \\
Protocol&(1-bit inputs)&&algebra&algebra\\
\hline

\end{tabular}
\caption{Comparison with the Yao's protocol, in the binary case}
\label{tab:comp}
\end{center}
\end{table}

\section{Intersecting Codes and Minimal Codes}
\label{sec:minimal}

In our protocol, we need linear codes where all codewords are \textit{minimal}. Let $C$ be a linear code of length $n$.  A codeword $c$ is said to be \textit{minimal} if $\forall c' \in C, (supp(c') \subset supp(c)) \implies$ ($c$ and $c'$ are linearly dependent). We say that a linear code $C$ is \textit{minimal} if every nonzero codeword of $C$ is minimal. This notion is closely related to the notion of intersecting codes \cite{CL85}. The notions are identical in the binary case but no more in the $q$-ary case (a minimal code is intersecting, but the inverse is not always true). We recall that an \textit{intersecting code} $C$ is such that for all nonzero $c,c' \in C, supp(c) \cap supp(c') \neq \emptyset$.

Interestingly, use of intersecting codes or minimal codewords has been suggested for oblivious transfer \cite{BCS96} and for secret sharing \cite{AB98,DY03,SL12}, which is a tool widely used for Secure Multi-Party Computation \cite{CDN12}.

\subsection{The binary case}

Due to the coincidence with the notion of intersecting codes, binary minimal codes have received a lot of attention \cite{CL85,Slo93,CZ94,BCS96,EC99}. For instance, \cite{CL85} gives definitions, some generic constructions and non-constructive bounds on rates; \cite{Slo93} gives explicit constructions for small dimensions and summarizes bounds on minimal distance; \cite{CZ94} gives an explicit constructive sequence of intersecting codes with high rate, and so on. We do not here detail these results. We only sum up what is important for us: there exist explicit constructions of minimal binary linear codes with good rates. 
Thus, our protocol of Section~\ref{sec:protocol} can be constructed in the binary case using codewords whose size grows linearly with the size of the inputs

\subsection{The $q$-ary case}

Finding minimal $q$-ary codes has received little attention \cite{DY03,GLL10,SL12} in the domain of secret sharing. \cite{SL12} details some properties of minimal linear codes, in particular some sufficient conditions for a code to be minimal are given. \cite{DY03,SL12} exhibit constructions of minimal codes using irreducible cyclic codes, which unfortunately do not achieve good rates.
As said before, simplex codes are minimal, they however suffer from a very bad rate. Indeed, a simplex code of dimension $k$ has length $q^{k}-1$. This gives us an existential and constructive result about $q$-ary minimal linear codes, but we still need better rates.

One can also build a $q$-ary minimal linear code by expanding the columns of the generator matrix of a binary intersecting codes and adding every column of elements in $\ff_q$ sharing the same support. This however does not lead to good codes either. For instance,  we can expand the $4\times 9$  $H$ matrix of Section~\ref{sec:example} to a $4 \times (4q+(q-1)^3)$ matrix, which is slightly better than the simplex code. For instance, with $q=3$, we obtain the following $4 \times 20$ matrix.


\[
H=
\begin{pmatrix}
1&0& 11& 0& 0 &00& 11& 00& 11111111\\
0&1& 12& 0& 0 &00& 00& 11& 11112222\\
0&0& 00& 1& 0 &11& 12& 00& 11221122\\
0&0& 00& 0& 1& 12& 00& 12& 12121212
\end{pmatrix}
\]

Interestingly, through other means, Song and Li \cite{SL12} also construct a $[20,4]$ ternary minimal code.

We exhibit two bounds on the rates of minimal codes. Unfortunately, these proofs are not constructive. 

\begin{theorem}[Maximal Bound]
Let $C$ a minimal linear $[n,k,d]$ $q$-ary code, then $R \leq \log_q(2)$
\end{theorem}
\begin{proof}
This bound is even true for non-linear minimal codes. Let us consider the family $F$ of the supports of the vectors of $C$. Due to the definition of minimal codes, this is a Sperner family. It is known that $|F|\leq {n \choose n/2}$. Thus, $|C| =q^k \leq 1+ (q-1) {n \choose n/2}$ then $R=k/n \leq \log_q (2)$.
\end{proof}

\begin{theorem}[Minimal Bound]
\label{thm:minbound}
For any $R$, $0 \leq R=k/n \leq \frac{1}{2} \log_q (\frac{q^2}{q^2-q+1})$, there exists an infinite sequence of $[n,k]$ minimal linear codes.
\end{theorem}
\begin{proof}
The proof is similar to the one of \cite{CL85} in the binary case. Let us fix $n$ and $k$. For $a \in \ff_q^n$, such that $|supp(a)|=i$, there are $q^i -q$ linearly independent vectors $b$ such that $supp(b) \subset supp(a)$. The pair $(a,b)$ belongs to $\begin{bmatrix}
n-2\\k-2
\end{bmatrix}$ linear $[n,k]$ codes, where $\begin{bmatrix}
x\\k
\end{bmatrix}$ denotes the $q$-ary Gaussian binomial coefficient.

There are less than $\sum\limits_{i=0}^n (q-1)^i (q^i-q) = (1+(q-1)q)^n - q^n \leq (q^2 - q +1)^n$ such ordered ``bad''  $(a,b)$ pairs. At least $\begin{bmatrix}
n\\k
\end{bmatrix} - \begin{bmatrix}
n-2\\k-2
\end{bmatrix} (q^2 - q +1)^n$ linear $[n,k]$ codes thus contain no ``bad'' pairs, \ie are minimal. For $k/n \leq \frac{1}{2} \log_q (\frac{q^2}{q^2-q+1})$, the quantity is positive.
\end{proof}

Notice that the minimal bound exposed in Theorem~\ref{thm:minbound} meets the $\frac{1}{2} \log_2(\frac{4}{3})$ bound in the binary case exhibited in \cite{CL85}. We can however not use the same techniques as in the binary case (e.g. \cite{CL85,CZ94}) to obtain explicit constructions with high rates, which remains an open issue.

\section{Conclusion}
\label{sec:conclu}

We present a theoretical protocol for performing secure two-party computation of linear functions based on linear codes and oblivious transfer only, using a parallel with the Wire-Tap Channel II model. Due to the efficiency of linear algebra and current constructions of oblivious transfer, this could be a basis for efficient protocols for secure evaluation of some classes of functions. 

Several leads for future research are:\\
$\bullet$ Constructions of good $q$-ary minimal linear codes;\\
$\bullet$ Other encoding techniques than linear coset coding;\\
$\bullet$ Techniques to encompass secure computation of non-linear functions;\\
$\bullet$ Techniques to deal with malicious adversaries.

\bibliographystyle{alpha}
\bibliography{arxiv13_CCP}

\begin{thebibliography}{CDN12}

\bibitem[AB98]{AB98}
Alexei~E. Ashikhmin and Alexander Barg.
\newblock Minimal vectors in linear codes.
\newblock {\em IEEE Transactions on Information Theory}, 44(5):2010--2017,
  1998.

\bibitem[BCP13]{BCP13}
Julien Bringer, Herv{\'e} Chabanne, and Alain Patey.
\newblock Shade: Secure hamming distance computation from oblivious transfer.
\newblock In {\em Workshop on Applied Homomorphic Cryptogrpahy (WAHC)}, 2013.

\bibitem[BCS96]{BCS96}
Gilles Brassard, Claude Cr{\'e}peau, and Miklos Santha.
\newblock Oblivious transfers and intersecting codes.
\newblock {\em IEEE Transactions on Information Theory}, 42(6):1769--1780,
  1996.

\bibitem[Bon84]{Bon84}
Arrigo Bonisoli.
\newblock Every equidistant linear code is a sequence of dual hamming codes.
\newblock {\em Ars Combinatoria}, 18:181--186, 1984.

\bibitem[CDN12]{CDN12}
Ronald Cramer, Ivan Damgard, and Jesper~Buus Nielsen.
\newblock Secure multiparty computation and secret sharing - an information
  theoretic approach.
\newblock Book Draft, 2012.

\bibitem[CL85]{CL85}
G{\'e}rard~D. Cohen and Abraham Lempel.
\newblock Linear intersecting codes.
\newblock {\em Discrete Mathematics}, 56(1):35--43, 1985.

\bibitem[CLZ94]{CLZ94}
G{\'e}rard~D. Cohen, Simon Litsyn, and Gilles Z{\'e}mor.
\newblock Upper bounds on generalized distances.
\newblock {\em IEEE Transactions on Information Theory}, 40(6):2090--2092,
  1994.

\bibitem[CZ94]{CZ94}
G{\'e}rard~D. Cohen and Gilles Z{\'e}mor.
\newblock Intersecting codes and independent families.
\newblock {\em IEEE Transactions on Information Theory}, 40(6):1872--1881,
  1994.

\bibitem[DY03]{DY03}
Cunsheng Ding and Jin Yuan.
\newblock Covering and secret sharing with linear codes.
\newblock In Cristian Calude, Michael~J. Dinneen, and Vincent Vajnovszki,
  editors, {\em DMTCS}, volume 2731 of {\em Lecture Notes in Computer Science},
  pages 11--25. Springer DMTCS, 2003.

\bibitem[EC99]{EC99}
Sylvia~B. Encheva and G{\'e}rard~D. Cohen.
\newblock Constructions of intersecting codes.
\newblock {\em IEEE Transactions on Information Theory}, 45(4):1234--1237,
  1999.

\bibitem[GLL10]{GLL10}
Yu-juan GUO, Zhi-hui LI, and Hong LAI.
\newblock A novel dynamic and verifiable secret sharing scheme based on linear
  codes.
\newblock {\em Journal of Shaanxi Normal University (Natural Science Edition)},
  4:013, 2010.

\bibitem[Gol04]{Gol04}
Oded Goldreich.
\newblock {\em The Foundations of Cryptography - Volume 2, Basic Applications}.
\newblock Cambridge University Press, 2004.

\bibitem[HL10]{HL10}
Carmit Hazay and Yehuda Lindell.
\newblock {\em Efficient Secure Two-Party Protocols}.
\newblock Springer, 2010.

\bibitem[Lip]{Lip}
Helger Lipmaa.
\newblock Oblivious transfer or private information retrieval.
\newblock
  \url{http://www.cs.ut.ee/~lipmaa/crypto/link/protocols/oblivious.php}.

\bibitem[LP09]{LP09}
Yehuda Lindell and Benny Pinkas.
\newblock A proof of security of yao's protocol for two-party computation.
\newblock {\em J. Cryptology}, 22(2):161--188, 2009.

\bibitem[OW84]{OW84}
Lawrence~H. Ozarow and Aaron~D. Wyner.
\newblock Wire-tap channel {II}.
\newblock In {\em EUROCRYPT}, pages 33--50, 1984.

\bibitem[Rab81]{Rab81}
Michael~O. Rabin.
\newblock How to exchange secrets with oblivious transfer.
\newblock Technical Report TR-81, Aiken Computation Lab, Harvard University,
  1981.

\bibitem[Sch12]{Sch12}
Thomas Schneider.
\newblock {\em Engineering Secure Two-Party Computation Protocols - Design,
  Optimization, and Applications of Efficient Secure Function Evaluation}.
\newblock Springer, 2012.

\bibitem[SL12]{SL12}
Yun Song and Zhihui Li.
\newblock Secret sharing with a class of minimal linear codes.
\newblock {\em CoRR}, abs/1202.4058, 2012.

\bibitem[Slo93]{Slo93}
N.J.A. Sloane.
\newblock Covering arrays and intersecting codes.
\newblock {\em Journal of Combinatorics Designs}, 1:51--63, 1993.

\bibitem[SSW09]{SSW09}
Ahmad-Reza Sadeghi, Thomas Schneider, and Immo Wehrenberg.
\newblock Efficient privacy-preserving face recognition.
\newblock In Donghoon Lee and Seokhie Hong, editors, {\em ICISC}, volume 5984
  of {\em Lecture Notes in Computer Science}, pages 229--244. Springer, 2009.

\bibitem[Wei91]{Wei91}
Victor K.-W. Wei.
\newblock Generalized hamming weights for linear codes.
\newblock {\em IEEE Transactions on Information Theory}, 37(5):1412--1418,
  1991.

\bibitem[Wyn75]{Wyn75}
Aaron~D. Wyner.
\newblock The wire-tap channel.
\newblock {\em The Bell System Technical Journal}, 54(8):1355--1387, October
  1975.

\bibitem[Yao86]{Yao86}
Andrew Chi-Chih Yao.
\newblock How to generate and exchange secrets (extended abstract).
\newblock In {\em FOCS}, pages 162--167. IEEE Computer Society, 1986.

\end{thebibliography}

\end{document}